\documentclass{interact}

\usepackage{hyperref}
\usepackage{amsfonts}
\usepackage{graphicx}
\usepackage{float}
\usepackage{comment}
\usepackage{amsthm,amssymb}
\usepackage{amsmath}
\usepackage{xcolor}

\newtheorem{theorem}{Theorem}

\newtheorem{lemma}{Lemma}

\newtheorem{remark}{Remark}

\usepackage[natbibapa,nodoi]{apacite}

\begin{document}
\renewcommand{\BBAA}{\&}
\renewcommand{\BBAB}{\&}

\author{
\name{G.~Portilla\textsuperscript{a}, I.V.~Alexandrova\textsuperscript{b}\thanks{CONTACT I.V.~Alexandrova. Email: i.v.aleksandrova@spbu.ru}, S.~Mondié\textsuperscript{a} and A.P.~Zhabko\textsuperscript{b}}
\affil{\textsuperscript{a}Department of Automatic Control,
        CINVESTAV-IPN, 07360 Mexico D.F., Mexico;\\ \textsuperscript{b}Department of Control Theory,
        St. Petersburg State University,\\ 7/9 Universitetskaya nab., 199034, St. Petersburg, Russia}
}

\title{Estimates for solutions of homogeneous time-delay systems: Comparison of Lyapunov--Krasovskii and Lyapunov--Razumikhin techniques}

\maketitle

\begin{abstract}
In this contribution, the estimates for the response of time delay systems with nonlinear homogeneous right-hand side of degree strictly greater than one are constructed.  The existing results  obtained  via  the  Lyapunov--Razumikhin  approach  are reminded.  Their proofs, revisited in the appendix, lead to explicit expressions of the involved constants. Based on a recently introduced Lyapunov--Krasovskii functional and known estimates of the domain of attraction, we present new estimates of the system response. We compare both approaches and discuss the illustrative examples. 
\end{abstract}



\section{Introduction}
Stability analysis via the linear approximation, when it is nonsingular, is usually the first choice method for studying nonlinear systems \citep{halanay1966differential}.
This is true in particular for systems with delays that have received outstanding attention in the past decades \citep{gu2003stability}.
When the linear approximation is zero, it is possible to perform the analysis by the first homogeneous approximation. This strategy  has lead to significant contributions to the design of controllers  \citep{hermes1991homogeneous}, and to robust stability analysis \citep{rosier1992homogeneous}.\\
Authors usually favour the Lyapunov--Razumikhin framework for studying the homogeneous time delay systems \citep{efimov2014}. In particular, delay-independent stability of homogeneous systems is addressed in \citet{aleksandrov2012asymptotic,aleksandrov2014delay,efimov2016}. The approach was applied to the estimation of the system response in \citet{aleksandrov2012asymptotic} and of the attraction region in \citet{aleksandrov2014delay}. It indeed appears natural to use the Lyapunov function of the corresponding delay-free system when applying the Lyapunov--Razumikhin approach to the homogeneous delay system. This strategy has established a remarkable result: if the delay-free system is asymptotically stable, then the trivial solution of the homogeneous delay system is asymptotically stable for all delays when the homogeneity degree is strictly greater than one \citep{aleksandrov2012asymptotic}.\\
A Lyapunov--Krasovskii functional for systems of homogeneity degree strictly greater than one was recently introduced in
\cite{Voronezh, alexandrova2019lyapunov}. It is built from the homogeneous delay free system
Lyapunov function, following ideas of the so-called complete type
functionals approach \citep{kharitonov2013time}. It has been successfully used to present
necessary and sufficient conditions for local stability and to give estimates of the
region of attraction \citep{alexandrova2019lyapunov}.
\\
The same case of homogeneity degree strictly greater than one is addressed in this contribution.
Our aim is to find estimates of  the solutions via the  Lyapunov--Krasovskii functional introduced in \citet{Voronezh, alexandrova2019lyapunov} and to compare them with those obtained via the  Lyapunov--Razumikhin approach \citep{aleksandrov2012asymptotic,aleksandrov2016asymptotic}.
Comparison of the domains of attraction obtained via both approaches \citep{aleksandrov2014delay,alexandrova2019lyapunov} is also presented. It is worth clarifying that the previously mentioned results concern local asymptotic stability, hence so do the presented results. For global asymptotic stability, stronger conditions, as those presented in \cite{efimov2014}, are required. Existing results derived via the Razumikhin theorem are revisited in the appendix. In contrast with the original works \citep{aleksandrov2012asymptotic,aleksandrov2014delay,aleksandrov2016asymptotic}, explicit formulae for all the involved constants are given, thus allowing the comparison. \\
To the best of our knowledge, few studies address direct practical comparison of the Lyapunov--Krasovskii and the Lyapunov--Razumikhin approaches. This work may be considered as a case study in this direction. Moreover, it is often assumed that the Lyapunov--Krasovskii approach makes possible the presentation of quantitative estimates of the convergence rate, whereas the Lyapunov--Razumikhin method only gives qualitative bounds \citep{efimov2014automatica}. Our analysis shows that for homogeneous time delay systems this does not hold. The structure of the estimates obtained via both approaches is the same, and they both depend on the Lyapunov function of the original system with zero delay. This common feature allows us to compare the involved constants directly. Furthermore, since the Lyapunov--Krasovskii functionals in use are those corresponding to the necessary and sufficient stability conditions, the comparison is expected to be revealing, in line with the linear case where the use of exact, prescribed derivative, functionals may lead to obtention of tighter estimates. The conclusion is quite surprising: although the estimate of the domain of attraction may be less conservative for the Lyapunov--Krasovskii approach,
the estimate of the response obtained via Lyapunov--Razumikhin one turns out to be much closer to the system response.
\\
The contribution is organised as follows. In Section II, basic definitions and properties of homogeneous systems are
introduced. We also present the existing estimates of the domain of attraction and of convergence rate obtained via the Lyapunov--Razumikhin approach. Estimates of the solutions using the Lyapunov--Krasovskii approach are presented in Section III. Section IV is devoted to an illustrative example. The paper ends with discussion in Section V, where we compare the results of both approaches.
The proof for the estimates obtained via the Razumikhin approach is given in the appendix.\\
\textbf{Notation: } The space of $\mathbb{R}^n$ valued continuous functions on $[-h,0],$ which is endowed with the norm $ \|\varphi\|_h=\max_{\theta\in[-h,0]}\|\varphi(\theta)\|,$ is denoted  by $C([-h,0],\mathbb{R}^n)$. Here, $\|\varphi(\theta)\|$ stands for the Euclidean norm. We use $|\varphi|_h$ instead of $\|\varphi\|_h,$ if $\varphi$ admits scalar values. The solution and the restriction of the solution to the segment $[t-h,t]$ are respectively denoted as $x(t)$ and $x_t$. If the initial condition is important, we write $x(t,\varphi)$ and $x_t(\varphi)$.

\section{Previous results}
In this paper, we study a homogeneous time delay system of the form
\begin{equation} \label{eq:delay_system}
\dot{x}(t)=f(x(t),x(t-h)), 
\end{equation}
where $x(t)\in \mathbb{R}^n,$ $h>0$ is a constant delay, the vector function $f(x,y)$ is continuously differentiable and homogeneous of degree $\mu>1$, i.e. $$f(cx,cy)=c^\mu f(x,y)\quad \forall c>0,\quad \forall x,y\in \mathbb{R}^n.$$ The initial condition is
\begin{equation*} 
x(\theta)=\varphi(\theta),\quad\varphi\in C([-h,0],\mathbb{R}^n). 
\end{equation*}
Consider also the delay free system
\begin{equation} \label{eq:delay-free_system}
    \dot{z}(t)=f(z(t),z(t)). 
\end{equation} 
In the sequel, system \eqref{eq:delay-free_system} is assumed to be asymptotically stable.\\

\subsection{Properties of homogeneous systems}
A homogeneous function $f(x,y)$, $x,y\in\mathbb{R}^n$, admits a bound of the form  \citep{zubov1964methods}
\begin{gather}\label{eq:bound_function}
    \|f(x,y)\|\leq m\left(\|x\|^\mu+\|y\|^\mu\right),\quad\text{where}\\
    m\geq\max_{\|x\|^\mu+\|y\|^\mu=1}\|f(x,y)\|>0.\notag
\end{gather}
Its derivative consists of the homogeneous functions of degree $\mu-1,$ hence there exist constants \mbox{$m_1,m_2>0$} such that
\begin{gather}\label{eq:bound_deriva_x_function}
    \left\|\frac{\partial f(x,y)}{\partial x}\right\|\leq m_1\left(\|x\|^{\mu-1}+\|y\|^{\mu-1}\right),
\\
\label{eq:bound_deriva_y_function}
    \left\|\frac{\partial f(x,y)}{\partial y}\right\|\leq m_2\left(\|x\|^{\mu-1}+\|y\|^{\mu-1}\right). 
\end{gather}
Since the delay free system \eqref{eq:delay-free_system} is asymptotically stable, 
it is possible to find a twice continuously differentiable positive definite Lyapunov function $V(x)$ \citep{zubov1964methods,rosier1992homogeneous} which is homogeneous of degree $\gamma\geq 2,$ and thus admits bounds of the form
\begin{equation}\label{eq:bound_V}
    k_0\|x\|^\gamma\leq V(x)\leq k_1\|x\|^\gamma,
\end{equation}
where $k_0,k_1>0.$ Furthermore, the time derivative of $V(x)$ along the solutions of system \eqref{eq:delay-free_system} satisfies the equation 
\begin{equation}\label{eq:bound_dot_V}
    \left(\frac{\partial V(x)}{\partial x}\right)^T f(x,x)\leq -\mathrm{w}\|x\|^{\gamma+\mu-1},\quad \mathrm{w}>0.
\end{equation}
Since the components of $\dfrac{\partial V(x)}{\partial x}$ and $\dfrac{\partial^2 V(x)}{\partial x^2}$ are the homogeneous functions of degrees $\gamma-1$ and $\gamma-2\geq 0,$ respectively, there exist constants $k_2,k_3>0$ such that
\begin{equation}  \label{eq:bound_deriv_x_V}
\left\|\frac{\partial V(x)}{\partial x}\right\|\leq k_2\|x\|^{\gamma-1},\ \ \left\|\frac{\partial^2 V(x)}{\partial x^2}\right\|\leq k_3\|x\|^{\gamma-2}.
\end{equation}
In the considerations below, an arbitrary Lyapunov function $V(x)$ which possesses all the above mentioned properties in a neighborhood $\|x\|\leq H,$ $H>0,$ may be used.
\begin{theorem} \textup{\citep{aleksandrov2014delay}}.
Let $\mu>1.$ If system \eqref{eq:delay-free_system} is asymptotically stable, then the trivial solution of the delay system \eqref{eq:delay_system} is also asymptotically stable for any delay $h\geq 0.$
\end{theorem}

\subsection{ A review on existing estimates via Lyapunov--Razumikhin approach}
\label{Sect_R}
Here, we recall the estimates of the domain of attraction and of the system response obtained  via the Lyapunov--Razumikhin  approach in \citet{aleksandrov2012asymptotic,aleksandrov2014delay,aleksandrov2016asymptotic}, where the Lyapunov function of the delay free system \eqref{eq:delay-free_system} is used. For the sake of completeness and for comparison purposes, we  include the detailed proofs into the appendix, where we provide explicit formulae for all involved constants.

Given $\alpha>1,$ introduce the Razumikhin condition
\begin{equation}\label{razu_condition}
    V(x(\xi))<\alpha V(x(t)),\quad \forall\,\xi\in [t-2h,t], \quad\forall\, t\geq h.
\end{equation}
It is shown in \citet{aleksandrov2014delay} that there exist $\delta>0$ and $k_5=k_5(\delta)>0$ such that the time derivative of $V(x)$ satisfies
\begin{equation}
\label{eq:der_Razum}
    \frac{dV(x(t))}{dt}\leq  -k_5\|x(t)\|^{\gamma+\mu-1}
\end{equation}
along the solutions of system \eqref{eq:delay_system} which obey the Razumikhin condition \eqref{razu_condition} and $\|x_t\|_h\leq\delta.$ Define the values
\begin{gather*}
\kappa=\left(\frac{k_0}{k_1}\right)^\frac{1}{\gamma},\quad
K=\Bigl(1+(\mu-1)mh(\kappa \delta)^{\mu-1}\Bigr)^\frac{1}{\mu-1}.
\end{gather*}
\begin{theorem}\textup{\citep{aleksandrov2014delay}}. \label{th:attraction_LR}
Let $\Delta$ be the root of the equation 
\begin{equation}\label{eq:attraction_LR}
    \Delta+mh\Delta^\mu=\dfrac{\kappa \delta}{K}.
\end{equation}
If system \eqref{eq:delay-free_system} is asymptotically stable, then the set of initial functions satisfying $\|\varphi\|_h<\Delta$ 
is contained in the attraction region of the trivial solution of system~\eqref{eq:delay_system}.
\end{theorem}
\begin{remark}
It follows from the proof of Theorem~\ref{th:attraction_LR} that if $\|\varphi\|_h<\Delta,$ then \mbox{$\|x(t,\varphi)\|<\delta$} for any $t\geq0.$
\end{remark}
\begin{theorem}\textup{\citep{aleksandrov2012asymptotic,aleksandrov2016asymptotic}}. \label{th:estimation_LR}
If system \eqref{eq:delay-free_system} is asymptotically stable, then there exist $\tilde{c}_1,$ $\tilde{c}_2>0$ such that the solutions of system~\eqref{eq:delay_system} with $\|\varphi\|_h < \Delta,$ where $\Delta$ is the root of equation \eqref{eq:attraction_LR}, admit an estimate of the form 
\begin{equation}
\label{eq:final_est_LR}
    \|x(t,\varphi)\|\leq \frac{\tilde{c}_1\|\varphi\|_h}{(1+\tilde{c}_2\|\varphi\|_h^{\mu-1}t)^\frac{1}{\mu-1}},\quad t\geq 0.
\end{equation}
\end{theorem}
The constants $k_5,$ $\delta,$ $\tilde{c}_1,$ $\tilde{c}_2$ are specified in the Appendix.

\section{Estimates via Lyapunov--Krasovskii approach}
\label{Sect_LK}
The purpose of this section is to present the estimates for solutions via the Lyapunov--Krasovskii approach. 
\subsection{General Case}
We use the functional introduced in \citet{Voronezh,alexandrova2019lyapunov}: 
\begin{gather}\label{eq:functional}
v(\varphi)=V(\varphi(0))+\left(\left.\frac{\partial V(x)}{\partial x}\right|_{x=\varphi(0)}\right)^T \int_{-h}^{0}f(\varphi(0),\varphi(\theta))d\theta \\+\int_{-h}^{0}(\mathrm{w_1}+(h+\theta)\mathrm{w_2})\|\varphi(\theta)\|^{\gamma+\mu-1}d\theta.\notag
\end{gather}
Here, $\mathrm{w}_1, \mathrm{w}_2>0$ are such that $\mathrm{w}_0=\mathrm{w}-\mathrm{w}_1-h\mathrm{w}_2>0$. It is well known that the Lyapunov functional~\eqref{eq:functional} is a measure of the states of the system, and that its lower and upper bounds, which we present below, are crucial to find the domain of attraction and the estimates of the solutions. 

\begin{lemma} \label{lemma_1}
\textup{\citep{alexandrova2019lyapunov}} There exist $a_1,a_2>0$ such that functional \eqref{eq:functional} admits a lower bound of the form
\begin{equation}\label{eq:firs_lower_bound_v}
    v(\varphi)\geq a_1\|\varphi(0)\|^\gamma+a_2\int_{-h}^{0}\|\varphi(\theta)\|^{\gamma+\mu-1} d\theta
\end{equation}
in the neighbourhood $\|\varphi\|_h\leq \delta.$ Here, $\delta\in(0,H_1),$
\begin{gather}
\begin{split}
\label{eq:H1}
    H_1&=\left(\frac{k_0}{hk_2m(1+\chi^{-2\mu})}\right)^{\frac{1}{\mu-1}}, \\
a_1&=k_0-hk_2m(1+\chi^{-2\mu})\delta^{\mu-1},\notag\\ a_2&=\mathrm{w}_1-k_2m\chi^{2(\gamma-1)}.\notag
\end{split}
\end{gather}
The constant $\chi>0$ is chosen in such a way that $a_2>0.$
\end{lemma}
\begin{lemma}
\label{lemma_2}
There exist $b_1,b_2>0$ such that functional \eqref{eq:functional} admits an upper bound of the form
\begin{equation} \label{eq:second_upper_bound_v}
    v(\varphi)\leq b\left(\|\varphi(0)\|^{\gamma}+\int_{-h}^{0}\|\varphi(\theta)\|^{\gamma}d\theta\right),\quad\text{if}\quad \|\varphi\|_h\leq \delta.
\end{equation}
Here, $b=\max\{b_1,b_2\},\,$
$b_1=k_1+2hmk_2\delta^{\mu-1},$ $\, b_2=(mk_2+\mathrm{w}_1+h\mathrm{w}_2)\delta^{\mu-1}.$
\end{lemma}
\begin{proof}
We estimate each summand of the functional. First,
\begin{gather*} 
    I_1(\varphi) = V(\varphi(0))\leq k_1\|\varphi(0)\|^\gamma,\\
    I_3(\varphi)=\int_{-h}^{0}(\mathrm{w_1}+(h+\theta)\mathrm{w_2})\|\varphi(\theta)\|^{\gamma+\mu-1}d\theta\leq (\mathrm{w}_1+h\mathrm{w}_2)\int_{-h}^{0}\|\varphi(\theta)\|^{\gamma+\mu-1}d\theta.
\end{gather*}
Second, it follows from \eqref{eq:bound_deriv_x_V} and \eqref{eq:bound_function} that 
\begin{align*}
I_2(\varphi) &=  \left(\left.\frac{\partial V(x)}{\partial x}\right|_{x=\varphi(0)}\right)^T \int_{-h}^{0}f(\varphi(0),\varphi(\theta))d\theta \\ &\leq hk_2m\|\varphi(0)\|^{\gamma+\mu-1}+k_2m\|\varphi(0)\|^{\gamma-1}\!\!\!\int_{-h}^{0}\!\!\|\varphi(\theta)\|^{\mu}d\theta. 
\end{align*}
Using the inequality $d^p g^q\leq d^{p+q}+g^{p+q},$ $\,p,q\geq1,$ $\,d,g\geq0$, we have
\begin{equation*} 
I_2(\varphi)\leq 2hk_2m\|\varphi(0)\|^{\gamma+\mu-1}+k_2m\int_{-h}^{0}\|\varphi(\theta)\|^{\gamma+\mu-1}d\theta.
\end{equation*}
Combining all summands together, we obtain
\begin{equation*}
v(\varphi)\leq (k_1+2hmk_2\|\varphi(0)\|^{\mu-1})\|\varphi(0)\|^{\gamma}
+(mk_2+\mathrm{w}_1+h\mathrm{w}_2)\int_{-h}^{0}\|\varphi(\theta)\|^{\gamma+\mu-1}d\theta.
\end{equation*}
Taking into account that $\|\varphi\|_h\leq \delta$, we arrive at \eqref{eq:second_upper_bound_v}.
\end{proof}
It follows from the proof of Lemma~\ref{lemma_2} that $v(\varphi)$ also admits an upper bound of the form
\begin{gather}\label{eq:third_upper_bound_v}
    v(\varphi)\leq k_1\|\varphi(0)\|^\gamma + \beta \|\varphi\|_h^{\gamma+\mu-1},
\end{gather}
where $\beta=(2k_2 m +\mathrm{w}_1+h\mathrm{w}_2)h.$\\
Now, we present the time derivative of the functional~\eqref{eq:functional} along the solutions of system~\eqref{eq:delay_system}. Notice that the positive constants $\mathrm{w}_0, \mathrm{w}_1, \mathrm{w}_2$  play an essential role in achieving the above lower bound, and also in establishing below the negativity of the derivative, in line with complete type functionals for linear case \citep{kharitonov2013time}.
\begin{lemma} \textup{\citep{alexandrova2019lyapunov}}
\label{lemma_3}
The time derivative of functional $v(\varphi)$ along the solutions of system~\eqref{eq:delay_system} satisfies
\begin{equation}\label{eq:sec_upper_bound_dot_v}
\frac{dv(x_t)}{dt}\leq -c\left(\|x(t)\|^{\gamma+\mu-1}+\int_{-h}^{0}\|x(t+\theta)\|^{\gamma+\mu-1}d\theta\right),
\end{equation}
if $\|x_t\|_h\leq\delta.$ Here, $c=\min\{c_1,c_2\},\,$ $c_1=\mathrm{w}_0-4hL\delta^{\mu-1},$
$\,c_2=\mathrm{w}_2-2L\delta^{\mu-1},$ \mbox{$\,L=mm_1k_2+m^2k_3,$} and $\delta\in(0,H_2),$ where 
\begin{equation*} 
H_2 = \left(\min\left\{\frac{\mathrm{w}_0}{4hL},\frac{\mathrm{w}_1}{2hL},\frac{\mathrm{w}_2}{2L}\right\}\right)^\frac{1}{\mu-1}.
\end{equation*}
\end{lemma}

Lemmas~\ref{lemma_1} and \ref{lemma_2} allow proving the following result on the estimate of the domain of attraction \citep{alexandrova2019lyapunov}.
\begin{theorem}\textup{\citep{alexandrova2019lyapunov}}
\label{thm:attr_region_LK}
Let $\Delta$ be a positive root of equation
\begin{equation}
\label{eq:attraction_LK}
k_1\Delta^\gamma + \beta \Delta^{\gamma + \mu -1} = a_1\delta^\gamma.
\end{equation}
If system \eqref{eq:delay-free_system} is asymptotically stable, then the set of initial functions $\|\varphi\|_h<\Delta$
is the estimate of the attraction region of the trivial solution of \eqref{eq:delay_system}.
\end{theorem}
\begin{remark}
It follows from the proof of Theorem~\ref{thm:attr_region_LK} that if $\|\varphi\|_h<\Delta,$ then $\|x(t,\varphi)\|< \delta$ for any $t\geq 0$.
\end{remark}

Next, we connect the functional $v(x_t)$ to its time derivative with the help of the above bounds, using the following technical result.
\begin{lemma} \label{lemma_inequality}
Let $p$ and $q$ be such that $p>q\geq 1$. Then, the following inequality is satisfied
\begin{equation} \label{inqx}
    \left(\|x(t)\|^{q}+\int_{-h}^{0}\|x(t+\theta)\|^{q}d\theta\right)^\frac{p}{q}   
\leq L_1\left(\|x(t)\|^{p}+\int_{-h}^{0}\|x(t+\theta)\|^{p}d\theta\right),
\end{equation}
where $L_1=\bigl(2\max\{1,h\}\bigr)^{\frac{p}{q}-1}$.
\end{lemma}
\begin{proof}
Given $d,g\geq 0$ and $l\geq 1$, the following inequality holds by convexity:
\begin{equation*}
    (d+g)^l\leq 2^{l-1}(d^l+g^l).
\end{equation*}
Applying this inequality to the left-hand side of \eqref{inqx},
we get  
\begin{equation} \label{inqx2}
\left(\|x(t)\|^{q}+\int_{-h}^{0}\|x(t+\theta)\|^{q}d\theta\right)^\frac{p}{q}
 \leq 2^{\frac{p}{q}-1}\left(\|x(t)\|^{p}+\left(\int_{-h}^{0}\|x(t+\theta)\|^{q}d\theta\right)^{\frac{p}{q}}\right).
\end{equation}
We use H\"older's inequality to end the proof:
\begin{equation*}
    \left(\int_{-h}^{0}\|x(t+\theta)\|^q d\theta\right)^{\frac{p}{q}}\leq h^{\frac{p-q}{q}}\int_{-h}^{0}\|x(t+\theta)\|^p d\theta.
\end{equation*}
Now, the result follows from \eqref{inqx2}.
\end{proof}

\begin{lemma} \label{lemma_connect_v_dot_v}
The following inequality is satisfied
\begin{equation}\label{eq:connect_v_dot_v} 
    \frac{dv(x_t)}{dt} \leq -L_2v(x_t)^{\frac{\gamma+\mu-1}{\gamma}},\quad t\geq0, 
\end{equation}
along the solutions with $\|x_t\|_h\leq \delta<\min\{H_1,H_2\}.$ Here,
$$
L_2=\frac{c}{b^{\frac{\gamma+\mu-1}{\gamma}}L_1},
$$
where $b$ comes from the upper bound \eqref{eq:second_upper_bound_v}, $c$ comes from the bound for the derivative \eqref{eq:sec_upper_bound_dot_v} and $L_1$ comes from Lemma 1 with $p=\gamma+\mu-1$ and $q=\gamma$.
\end{lemma}
\begin{proof}
Combining \eqref{eq:sec_upper_bound_dot_v} with Lemma \ref{lemma_inequality}, we get
\begin{align*}
&\frac{dv(x_t)}{dt}\leq -c\left(\|x(t)\|^{\gamma+\mu-1}+\int_{-h}^{0}\|x(t+\theta)\|^{\gamma+\mu-1}d\theta\right)\\
&\leq-\frac{c}{L_1b^{\frac{\gamma+\mu-1}{\gamma}}}\left[b\left(\|x(t)\|^{\gamma}+\int_{-h}^{0}\|x(t+\theta)\|^{\gamma}d\theta\right)\right]^{\frac{\gamma+\mu-1}{\gamma}}.
\end{align*}
The lemma now follows from the upper bound \eqref{eq:second_upper_bound_v}.
\end{proof}

By using the comparison lemma \citep{khalil1996nonlinear}, we are finally able to present estimates for the solutions.
\begin{theorem} \label{th:estimation_LK}
Let system \eqref{eq:delay-free_system} be asymptotically stable. The solutions of system \eqref{eq:delay_system} with initial functions satisfying \mbox{$\|\varphi\|_h< \Delta,$} where $\Delta$ is defined in \eqref{eq:attraction_LK}, admit an estimate of the form 
\begin{equation}
\label{est_LKr_final}
    \|x(t,\varphi)\| \leq \hat{c}_1 \|\varphi\|_h\left[1+\hat{c}_2\|\varphi\|_h^{\mu-1}t\right]^{-\frac{1}{\mu-1}},\quad t\geq 0,
\end{equation}
where 
\begin{align*}
\hat{c}_1&=\left(\frac{k_1+\beta \Delta^{\mu - 1}}{a_1}\right)^\frac{1}{\gamma}=\frac{\delta}{\Delta},\\ \hat{c}_2&=\frac{c}{b}\left(\frac{\mu-1}{\gamma}\right)\left(\frac{k_1+\beta \Delta^{\mu - 1}}{2b\max\{1,h\}}\right)^{\frac{\mu-1}{\gamma}}.
\end{align*}
\end{theorem}
\begin{proof}
Since $v(x_t)$ fulfils the differential inequality \eqref{eq:connect_v_dot_v}, we take $t_0=0$ and consider the comparison equation
\begin{equation}
\label{eq:comparison_LK}
    \frac{du(t)}{dt} = -L_2u^{\frac{\gamma+\mu-1}{\gamma}}
(t)
\end{equation}
with initial condition
\begin{equation*}
   u(0)=u_0=(k_1+\beta \Delta^{\mu - 1})\|\varphi\|_h^\gamma.
\end{equation*}
Such choice of $u_0$ implies $v(\varphi)< u_0$ due to \eqref{eq:third_upper_bound_v}, if \mbox{$\|\varphi\|_h<\Delta.$} The solution of equation \eqref{eq:comparison_LK} can be easily obtained by using the method of separation of variables:
\begin{equation*}
    u(t) = u_0\left[1+L_2\left(\frac{\mu-1}{\gamma}\right)u_0^{\frac{\mu-1}{\gamma}}t\right]^{-\frac{\gamma}{\mu-1}}.
\end{equation*}
Hence, using the comparison lemma and bound \eqref{eq:firs_lower_bound_v}, we get
\begin{equation*}
   a_1\|x(t,\varphi)\|^\gamma\leq v(x_t) \leq u_0\left[1+L_2\left(\frac{\mu-1}{\gamma}\right)u_0^{\frac{\mu-1}{\gamma}}t\right]^{-\frac{\gamma}{\mu-1}}.
\end{equation*}
The result follows from the previous inequality.
\end{proof}
\begin{remark}
The structure of the estimates \eqref{eq:final_est_LR} and \eqref{est_LKr_final} is the same, and it is similar to that of a counterpart for delay free homogeneous systems \citep{zubov1964methods}.
\end{remark}

\subsection{Scalar Case}
In this section, we present less conservative estimates for the solutions of
the scalar equation of the form
\begin{equation}\label{scalar}
    \dot{x}(t)=\alpha_1 x^\mu(t)+\alpha_2 x^\mu(t-h),
\end{equation}
where the constants $\alpha_1,\alpha_2 \in \mathbb{R}$, $x\in\mathbb{R}$, $h\in\mathbb{R}^+$ and $\mu>1$ is an odd entire number. Assume that $\alpha_1+\alpha_2<0$ which implies the asymptotic stability of the trivial solution of \eqref{scalar}. Take $\mathrm{w}=-2(\alpha_1+\alpha_2)>0,$ $V(x)=x^2,$ and use the following modification of functional \eqref{eq:functional} presented in \citet{alexandrova2019lyapunov}:
\begin{equation}\label{scalar_functional}
v(\varphi)=\left(\varphi(0) + \alpha_2\int_{-h}^0 \varphi^\mu(\theta)d\theta \right)^2+\int_{-h}^{0}(\mathrm{w_1}+(h+\theta)\mathrm{w_2})\varphi^{\mu+1}(\theta)d\theta.
\end{equation}
Here, $\mathrm{w}_1, \mathrm{w}_2>0$ are such that $\mathrm{w}_0=\mathrm{w}-\mathrm{w}_1-h\mathrm{w}_2>0$. Bounds \eqref{eq:firs_lower_bound_v}, \eqref{eq:second_upper_bound_v} and \eqref{eq:sec_upper_bound_dot_v} for functional \eqref{scalar_functional} take the form
\begin{align*}
    v(\varphi)&\geq a_1\varphi^2(0)+a_2\int_{-h}^0\varphi^{\mu+1}(\theta)d\theta,\quad |\varphi|_h\leq \delta,\\
    v(\varphi)&\leq b\left(\varphi^2(0)+\int_{-h}^0\varphi^2(\theta)d\theta\right),\quad |\varphi|_h\leq \delta,\notag\\
    \frac{dv(x_t)}{dt}&\leq -c\left(x^{\mu+1}(t)+\int_{-h}^{0}x^{\mu+1}(t+\theta)d\theta\right),\quad |x_t|_h\leq \delta.\notag
\end{align*}
Here, $\delta<\min\{H_1,H_2\},$ and
\begin{align*}
H_1&=\left(\dfrac{\mathrm{w}_1\chi^{2}}{|\alpha_2|}\right)^{\frac{1}{\mu-1}},\quad H_2 = \left(\min\left\{\dfrac{\mathrm{w}_0}{hL},\dfrac{\mathrm{w}_2}{L}\right\}\right)^\frac{1}{\mu-1},
    \quad 0<\chi<\sqrt{\dfrac{1}{h|\alpha_2|}},\\ L&=|\alpha_2||\alpha_1+\alpha_2|,\quad
    a_1=1-\chi^2 h |\alpha_2|>0,\quad a_2=\mathrm{w}_1 - |\alpha_2|\chi^{-2}\delta^{\mu-1}>0,\\
b&=\max\{1+|\alpha_2| h;\bigl(|\alpha_2|(1+|\alpha_2|h)\delta^{\mu-1}+\mathrm{w}_1+h\mathrm{w}_2\bigr)\delta^{\mu-1}\},\\
c&=\min\{\mathrm{w}_0-hL\delta^{\mu-1};\mathrm{w}_2-L\delta^{\mu-1}\}.
\end{align*}

The functional admits also an upper bound
\begin{gather*}
    v(\varphi)\leq \varphi^2(0) +\beta |\varphi|_h^{\mu+1},
\end{gather*}
where $\beta=(2|\alpha_2|+ \alpha_2^2h \delta^{\mu-1} + \mathrm{w}_1+h\mathrm{w}_2)h$. 
As in Theorem~\ref{thm:attr_region_LK}, if $\Delta$ is a positive root of equation
\begin{equation}\label{eq:Delta_scalar}
    \Delta^2 +\beta \Delta^{\mu+1}=a_1\delta^2,
\end{equation}
then the set of initial functions $|\varphi|_h<\Delta$
is the estimate of the region of attraction of the trivial solution of \eqref{scalar}. Now, we present another
technical lemma, which is less conservative than Lemma~\ref{lemma_inequality}.
\begin{lemma} \label{lemma_inequality_scalar}
Let $k\geq2$ be an entire number. Then,
\begin{equation} \label{inqx_scalar}
    \left(x^2(t)+\int_{-h}^{0}x^2(t+\theta)d\theta\right)^k  
\leq L_1\left(x^{2k}(t)+\int_{-h}^{0}x^{2k}(t+\theta)d\theta\right),
\end{equation}
where $L_1=2^{k-2}(1+h)^{k-1}$.
\end{lemma}
\begin{proof}
The proof is carried out by mathematical induction. It is easy to verify \eqref{inqx_scalar} for $k=2$ taking into account that
$$
\left(\int_{-h}^{0}x^2(t+\theta)d\theta\right)^2\leq h\int_{-h}^{0}x^4(t+\theta)d\theta.
$$
Assume that \eqref{inqx_scalar} holds for a $k\geq 2,$ then
\begin{multline*}
    \left(x^2(t)+\int_{-h}^{0}x^2(t+\theta)d\theta\right)^{k+1}\leq 2^{k-2}(1+h)^{k-1}\\
    \times\left(x^{2k}(t)+\int_{-h}^{0}\!\!x^{2k}(t+\theta)d\theta\right)
    \!\!\left(x^2(t)+\int_{-h}^{0}\!\!x^2(t+\theta)d\theta\right).
\end{multline*}
Using $d^p g^q\leq d^{p+q}+g^{p+q}$ in the right-hand side of this inequality, we get
\begin{multline*} 
    \left(x^2(t)+\int_{-h}^{0}x^2(t+\theta)d\theta\right)^{k+1}\leq 2^{k-1}(1+h)^{k}\left(x^{2(k+1)}(t)+\int_{-h}^{0}x^{2(k+1)}(t+\theta)d\theta\right),
\end{multline*}
 and the lemma is proved.
\end{proof}
Using Lemma~\ref{lemma_inequality_scalar} with $k=(\mu+1)/2\in\mathbb{Z}$ instead of Lemma~\ref{lemma_inequality} and repeating the steps of the previous section, we arrive at the following estimates for solutions of \eqref{scalar}.
\begin{theorem} \label{th:estimation_LK_scalar}
If $\alpha_1+\alpha_2<0,$ then
the solutions of equation \eqref{scalar} with initial functions satisfying \mbox{$|\varphi|_h< \Delta,$} where $\Delta$ is a positive root of \eqref{eq:Delta_scalar}, admit an estimate of the form 
\begin{equation}
\label{est_LKr_final_scalar}
    |x(t,\varphi)| \leq \hat{c}_1 |\varphi|_h\left[1+\hat{c}_2|\varphi|_h^{\mu-1}t\right]^{-\frac{1}{\mu-1}},
\end{equation}
where 
\begin{align*}
\hat{c}_1&=\left(\frac{1 +\beta \Delta^{\mu-1}}{a_1}\right)^\frac{1}{2}=\frac{\delta}{\Delta},\\ \hat{c}_2&=\frac{c\left(\mu-1\right)}{b}\left(\frac{1 +\beta \Delta^{\mu-1}}{2b(1+h)}\right)^{\frac{\mu-1}{2}}.
\end{align*}
\end{theorem}
\begin{remark}
The reduction of conservatism in Theorem~\ref{th:estimation_LK_scalar} in comparison with Theorem~\ref{th:estimation_LK} is due to use of Lemma~\ref{lemma_inequality_scalar} and the presence of the additional summand in functional \eqref{scalar_functional}.
\end{remark}

\section{Illustrative Examples}  
\subsection{Example 1}
Consider a scalar equation of the form 
\begin{equation}\label{example}
    \dot{x}(t)=\alpha_1 x^3(t)+\alpha_2 x^3(t-h),
\end{equation}
where the constants $\alpha_1,\alpha_2 \in \mathbb{R}$, $x\in\mathbb{R}$ and $h\in\mathbb{R}^+$. In this case the homogeneity degree is $\mu=3$.
Take the system parameters $(\alpha_1,\alpha_2)=(-1,0.5)$, delay $h=10,$ and $\mathrm{w}=1.$ Now, compute the constants $m=\max\{|\alpha_1|,|\alpha_2|\}=1,\ m_1=3|\alpha_1|=3,\ m_2=3|\alpha_2|=1.5,\ k_0=k_1=1,\ k_2=k_3=2.$\\
First, we estimate the region of attraction using both approaches. We apply Theorem~2 with $\delta=H$ for the Lyapunov--Razumikhin approach, and solve equation \eqref{eq:Delta_scalar} tuning the parameters therein for the Lyapunov--Krasovskii one. The parameters and the obtained estimates are shown in Table~\ref{table_LK_attrac_region}.   
\begin{table}[h]
\caption{Constants for the estimates of attraction region, Example~1}
\label{table_LK_attrac_region}
\begin{center}
\begin{tabular}{cccccccccc}
\hline
\multicolumn{10}{c}{Lyapunov--Krasovskii approach}\\
\hline
$\Delta$ & $H_1$ & $H_2$ & $\delta$ & $\chi$ & $a_1$ & $\beta$ & $\mathrm{w}_0$ & $\mathrm{w}_1$ & $\mathrm{w}_2$\\
\hline
$0.067$ & $0.1012$ & $0.3$ & $0.1011$ & $0.32$ & $0.48$ & $17.7$ & $0.25$ & $0.05$ & $0.07$\\
\hline
\multicolumn{10}{c}{Lyapunov--Razumikhin approach}\\
\hline
\multicolumn{3}{c}{$\Delta$} & \multicolumn{2}{c}{$H$} & \multicolumn{3}{c}{$\kappa$} & \multicolumn{2}{c}{$K$}\\
\hline
\multicolumn{3}{c}{$0.0427$} & \multicolumn{2}{c}{$0.0443$} & \multicolumn{3}{c}{$1$} & \multicolumn{2}{c}{$1.001$}\\
 \hline
\end{tabular}
\end{center}
\end{table}
Next, we turn our attention to the estimates of the solutions. Although the estimate of the region of attraction is found to be less conservative via the Lyapunov--Krasovskii approach, we take  the same $\delta$ in Theorems \ref{th:estimation_LR} and \ref{th:estimation_LK_scalar} for comparison purposes. The constants characterising the estimates are shown in Table \ref{table_LK}.
\begin{table}[h]
\caption{Constants for the estimates of solutions, Example~1}
\label{table_LK}
\begin{center}
\begin{tabular}{cccccccccc}
\hline
\multicolumn{10}{c}{Lyapunov--Krasovskii approach}\\
\hline
$\delta$ & $H_1$ & $H_2$ & $\Delta$ & $\hat{c}_1$ & $\hat{c}_2$ & $\chi$ & $\mathrm{w}_0$ & $\mathrm{w}_1$ & $\mathrm{w}_2$\\
\hline
$0.01$ & $0.015$ & $0.26$ & $0.0099$ & $1.0014$ & $4.2\cdot 10^{-5}$ & $0.015$ & $0.33$ & $0.5$ & $0.017$\\
\hline
\multicolumn{10}{c}{Lyapunov--Razumikhin approach}\\
\hline
\multicolumn{2}{c}{$\delta$} & $H$ & \multicolumn{2}{c}{$\Delta$} & $\tilde{c}_1$ & \multicolumn{2}{c}{$\tilde{c}_2$} & $\alpha$ & $\rho$\\
\hline
\multicolumn{2}{c}{$0.01$} & $0.0443$ & \multicolumn{2}{c}{$0.009$} & $1.002$ & \multicolumn{2}{c}{$0.95$} & $2$ & $0.94$\\
 \hline
\end{tabular}
\end{center}
\end{table}

The estimate of the system response obtained via the Lyapunov--Krasovskii approach is then
\begin{equation*}
    |x(t,\varphi)| \leq 1.0014 |\varphi|_{10}\left[1+4.2\cdot 10^{-5}|\varphi|_{10}^{2}t\right]^{-\frac{1}{2}},
\end{equation*}
whereas using the Lyapunov--Razumikhin approach we arrive at
\begin{equation*}
    |x(t,\varphi)| \leq  1.002|\varphi|_{10}\left[1+0.95|\varphi|_{10}^{2}t\right]^{-\frac{1}{2}}.
\end{equation*}
For the initial condition $\varphi(\theta)=0.009$, $\theta \in [-10,0]$, the system response and the estimates \eqref{eq:final_est_LR} and \eqref{est_LKr_final_scalar} are depicted in Figure~\ref{figure_estimation} and Figure~\ref{figure_estimation_2} as continuous, dashed lines and dashed-dot lines, respectively. Although the Lyapunov--Krasovskii approach gives a better estimate at the beginning of the response as shown in  Figure~\ref{figure_estimation_2}, the bound obtained via the Lyapunov--Razumikhin one is much tighter, see Figure~\ref{figure_estimation}.
\begin{figure}[h]
     \centering
      \includegraphics[width=0.8\textwidth]{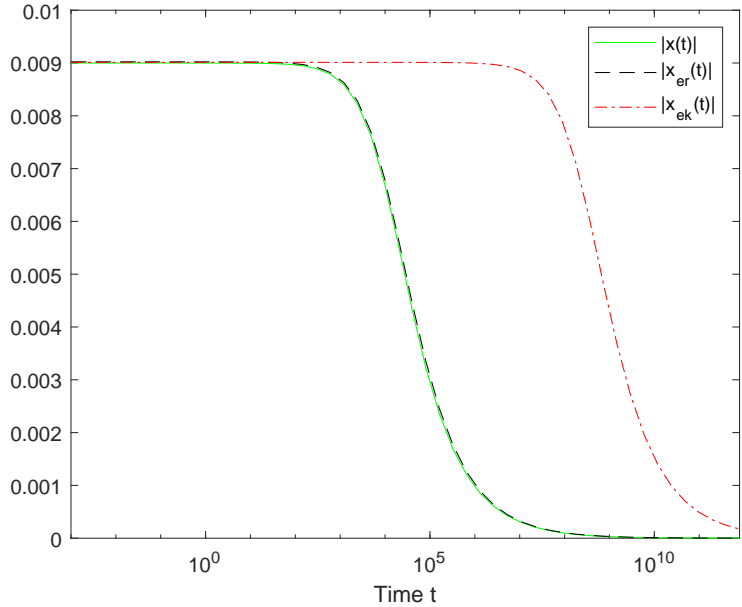}
      \caption{Estimates for the solution of \eqref{example} with a log scale for time}
      \label{figure_estimation}
  \end{figure}
  \begin{figure}[h]
     \centering
      \includegraphics[width=0.8\textwidth]{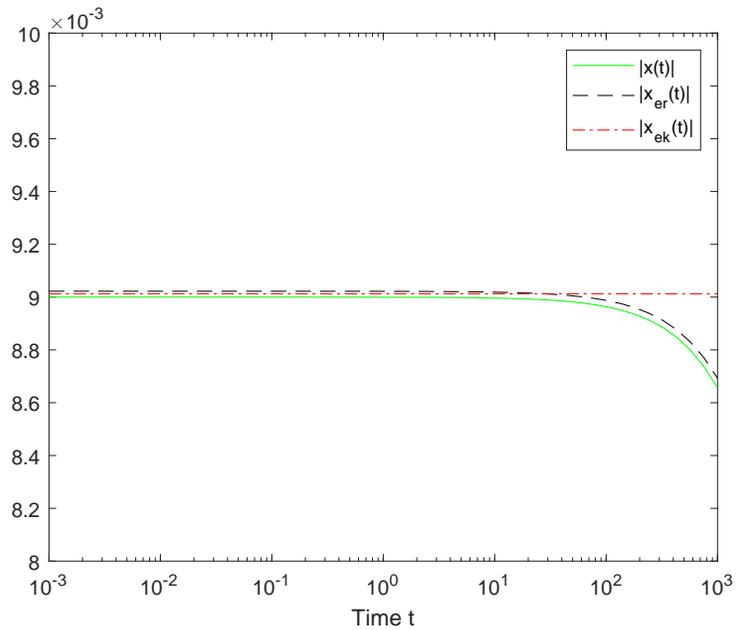}
      \caption{Estimates for the solution of \eqref{example} with a log scale for time}
      \label{figure_estimation_2}
  \end{figure}

\subsection{Example 2} 
Consider the system
\begin{align}
\dot{x}_1(t)&=x_2^\mu(t),\label{example_2}\\
\dot{x}_2(t)&=-x_1^\mu(t)-x_2^\mu(t-h),\notag
\end{align}
where $x_1(t),\ x_2(t)\in\mathbb{R}$ and $h\in\mathbb{R}^+$. Assume that the homogeneity degree $\mu$ of the right-hand side is a rational number with odd numerator and denominator, and $\mu>2.$ It is shown in \cite{book_in_Russian} that the delay-free system corresponding to \eqref{example_2} is asymptotically stable. The result is achieved with the help of the Lyapunov function
\begin{equation*}
    V(x)=\frac{1}{\mu+1}\left(x_1^{\mu+1}+x_2^{\mu+1}\right)+\zeta x_1^\mu x_2,
\end{equation*}
where a parameter $\zeta$ satisfies
\begin{equation*}
    0<\zeta< \min\left\{\frac{1}{\mu+1},\frac{4}{(\mu+1)^2}\right\}.
\end{equation*}
Furthermore, on the basis of computations in \cite{book_in_Russian} we obtain that along the solutions of \eqref{example_2} with $h=0,$
\begin{gather*}
    \frac{dV(x(t))}{dt}=-x_1^{2\mu}f\left(\frac{x_2}{x_1}\right),\quad\text{where}\quad f(z)=z^{2\mu}-\zeta\mu z^{\mu+1}+\zeta z^\mu+\zeta,
    \\
    f(z)\geq  \eta\left(z^{2\mu}+1\right),\quad \eta=\min\left\{\zeta;1-\zeta(\mu+1);\frac{\zeta}{1+\zeta}\left(1-\frac{\zeta(1+\mu)^2}{4}\right)\right\}.
\end{gather*}
Hence,
\begin{equation*}
    \frac{dV(x(t))}{dt}\leq -\eta\left(x_1^{2\mu}+x_2^{2\mu}\right)\leq -\mathrm{w}\|x\|^{2\mu},\quad\text{where}\quad \mathrm{w}=\frac{\eta}{2^{\mu-1}}.
\end{equation*}
First, calculate the necessary constants:
\begin{gather*}
m=\sqrt{2},\quad m_1=m_2=\mu,\quad k_0=(1/2)^\frac{\mu-1}{2}\left(\frac{1}{\mu+1}-\zeta\right),\quad k_1=\left(\frac{1}{\mu+1}+\zeta\right),\\
 k_2=\sqrt{(1+\zeta\mu)^2+(1+\zeta)^2},\quad k_3=\mu(1+\zeta\mu).
\end{gather*}
For the system parameters $\mu=3$, $\zeta=0.1$, $\mathrm{w}=0.0136$ and $h=1$, the constants characterising the estimates for the solutions of system \eqref{example_2} are shown in Table \ref{table_LK_2}.
\begin{table}[h]
\caption{Constants for the estimates of solutions, Example~2}
\label{table_LK_2}
\begin{center}
\begin{tabular}{cccccccccc}
\hline
\multicolumn{10}{c}{Lyapunov--Krasovskii approach}\\
\hline
$\delta$ & $H_1$ & $H_2$ & $\Delta$ & $\hat{c}_1$ & $\hat{c}_2$ & $\chi$ & $\mathrm{w}_0$ & $\mathrm{w}_1$ & $\mathrm{w}_2$\\
\hline
$0.001$ & $0.0111$ & $0.0058$ & $6.789\cdot 10^{-4}$ & $1.4728$ & $0.002$ & $0.39$ & $0.0022$ & $0.0092$ & $0.0022$\\
\hline
\multicolumn{10}{c}{Lyapunov--Razumikhin approach}\\
\hline
\multicolumn{2}{c}{$\delta$} & $H$ & \multicolumn{2}{c}{$\Delta$} & $\tilde{c}_1$ & \multicolumn{2}{c}{$\tilde{c}_2$} & $\alpha$ & $\rho$\\
\hline
\multicolumn{2}{c}{$0.001$} & $0.0066$ & \multicolumn{2}{c}{$6.8\cdot 10^{-4}$} & $1.4698$ & \multicolumn{2}{c}{$0.019$} & $2$ & $0.0643$\\
 \hline
\end{tabular}
\end{center}
\end{table}
For the initial condition $\varphi(\theta)=[4.8\cdot 10^{-4},\ 4.8\cdot 10^{-4}]$, $\theta \in [-1,0]$, the system response and the estimates obtained via the Lyapunov--Razumukhin and Lyapunov--Krasovskii approaches are depicted in Figure~\ref{figure_estimation_ex_2} as a continuous, dashed and dashed-dot line, respectively. We conclude that the former estimate is closer to the system response than the latter one.
  \begin{figure}[h]
     \centering
      \includegraphics[width=0.8\textwidth]{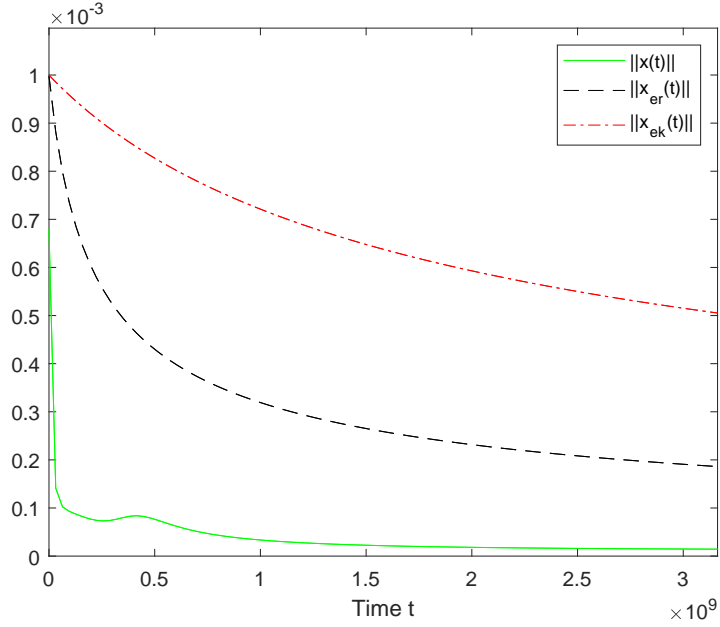}
      \caption{Estimates for the solution of system \eqref{example_2}}
      \label{figure_estimation_ex_2}
  \end{figure} 

\section{Discussion}
Unexpectedly, we obtain that the Lyapunov--Razumikhin approach gives much better results for homogeneous systems. Although in the first example, tuning some parameters of the Lyapunov--Krasovskii approach, we are able to achieve a less conservative bound for the region of attraction, as well as the closer estimate at the beginning of the response, this is not the case in general. Using the same value for $\delta$ in both approaches,
let us compare the estimates for the solutions \eqref{eq:final_est_LR} and \eqref{est_LKr_final}.
To this end, compare the values $\hat{c}_1$ and $\hat{c}_2$ with $A$ and $B$ (see the proof of Theorem~\ref{th:estimation_LR} in the Appendix) that precede $\tilde{c}_1$ and $\tilde{c}_2.$ 
While $\hat{c}_1$ and $A$ admit the same expression $\delta/\Delta,$ the values $\hat{c}_2$ and $B$ can be presented in the form
\begin{align*}
\hat{c}_2&=\frac{c}{bL_1}\left(\frac{\mu-1}{\gamma}\right)\left(\frac{a_1}{b}\right)^{\frac{\mu-1}{\gamma}}\left(\frac{\delta}{\Delta}\right)^{\mu-1},\\
B&=\rho\left(\frac{\mu-1}{\gamma}\right)k_0^{\frac{\mu-1}{\gamma}}\left(\frac{\delta}{\Delta}\right)^{\mu-1}.
\end{align*}
As the values $\delta$ and $\Delta$ are usually rather small in practice, conditions \eqref{firs_condition} and \eqref{second_condition} become insignificant, and one can often take $\rho=\overline{d}-\varepsilon,$ where $\varepsilon>0$ is arbitrarily small.
In this case, the constant
$B$ is close to
$$
\widetilde{B} = \frac{k_5}{k_1}\left(\frac{\mu-1}{\gamma}\right)\left(\frac{k_0}{k_1}\right)^{\frac{\mu-1}{\gamma}}\left(\frac{\delta}{\Delta}\right)^{\mu-1}.
$$
Clearly, the values $k_5,$ $k_1,$ and $k_0$ have a similar meaning for the Lyapunov function as the values $c,$ $b,$ and $a_1$ for the Lyapunov--Krasovskii functional.
The analysis of how the constants are built gives the expressions
\begin{gather*}
    a_1=k_0-\textup{const}\cdot\delta^{\mu-1},\quad
    b=\max\{k_1 + \textup{const}\cdot\delta^{\mu-1},\textup{const}\cdot\delta^{\mu-1}\},\\
    c=\min\{\textup{w}_0-\textup{const}\cdot\delta^{\mu-1},\textup{w}_2-\textup{const}\cdot\delta^{\mu-1}\},\quad k_5=\textup{w} - \textup{const}\cdot\delta^{\mu-1},
\end{gather*}
where ``$\textup{const}$'' means a different (positive) constant in each case. One can notice that $a_1<k_0,$ $b>k_1.$ Although it is not easy to compare the constants involved in the expressions for $c$ and $k_5$ directly, $\textup{w}>\textup{w}_0$ and $\textup{w}>\textup{w}_2$ imply that $k_5>c$ for a sufficiently small $\delta.$ Hence, if $\delta$ is sufficiently small, then
$$
\frac{k_5}{k_1}>\frac{c}{b},\quad \frac{k_0}{k_1}>\frac{a_1}{b}.
$$
In addition, $L_1>1.$
We arrive at the following conclusions:
\begin{itemize}
    \item[$\star$] First, $\hat{c}_2<\widetilde{B}\approx B<\tilde{c}_2$ at least when $\delta$ is rather small. Second, $\hat{c}_1<A<\tilde{c}_1,$ if $\Delta$ for the Lyapunov--Krasovskii approach is larger than those for the Lyapunov--Razumikhin one. If this is the case, then we are able to get a closer estimate at the beginning of the system response via the Lyapunov--Krasovskii method as in Example~1, see Figure~\ref{figure_estimation_2}. However, this is not the case in general. Moreover, since the constant in the denominator is in fact the one accounting for the rate of convergence, we conclude that the Lyapunov--Razumikhin approach gives in general a tighter bound, see Figures~\ref{figure_estimation} and \ref{figure_estimation_ex_2}.
    \item[$\star$] The sources of conservatism for the Lyapunov--Krasovskii approach are $(a)$ the multiplier $L_1 >1$ appearing due to the fact that the degrees of the upper bound for the functional and the bound for its derivative along the solutions are different, and $(b)$ the constants $c,$ $b,$ and $a_1$ for the Lyapunov--Krasovskii functional in comparison with $k_5,$ $k_1,$ and $k_0$ for the Lyapunov function.
    \item[$\star$] The sources of conservatism for the Lyapunov--Razumikhin approach are $(c)$ additional conditions \eqref{firs_condition} and \eqref{second_condition} on $\rho$ appearing due to the fact that the solution of the comparison equation should satisfy the Razumikhin condition, and $(d)$ the fact that the bound is first obtained for $t\geq h$ only, and that it is necessary to supplement it with the bound for $t\in [0,h].$
    \item[$\star$] For the scalar equation \eqref{scalar}, the Lyapunov--Krasovskii bound can be slightly improved (compare \eqref{est_LKr_final_scalar} to \eqref{est_LKr_final}) by introducing a less conservative multiplier $L_1>1.$
    \item[$\star$] Assuming that $V(x),$ and hence the constants $k_j,$ $j=0,\ldots,5,$ depend on $\mathrm{w}$ linearly, what is natural in view of \eqref{eq:bound_dot_V}, we obtain that the Lyapunov--Razumikhin bounds do not depend on $\mathrm{w}.$ At the same time, the Lyapunov--Krasovskii ones depend on the relation between $\mathrm{w}_0,$ $\mathrm{w}_1$ and $\mathrm{w}_2,$ hence these parameters may serve for optimization purposes.
\end{itemize}

The sources of conservatism for the Lyapunov--Razumikhin approach turn out to be insignificant. Conditions \eqref{firs_condition} and \eqref{second_condition} hold in the example because $\delta$ is small. As for $(d)$, the same bound obtained for $t\geq h$ holds for all $t\geq0$ in general. We conclude that the Lyapunov function of the delay free system works rather well, and better than the Lyapunov--Krasovskii functional, in the estimation of the convergence rate of a homogeneous time delay system solutions. This unexpected conclusion may be proceeding from the delay-independent stability property of the systems under consideration.
We also notice that slightly different assumptions on the right-hand sides are exploited in Sections~\ref{Sect_R} and \ref{Sect_LK}: differentiability of $f(x,y)$ with respect to $y$ in the Lyapunov--Razumikhin approach and with respect to $x$ in the Lyapunov--Krasovskii one.

\section{Conclusion}
In  this  paper,  we  obtain  estimates  of  the  response  of  nonlinear  homogeneous systems with delay and  right-hand side homogeneity degree strictly greater than one. The problem is addressed via the Lyapunov--Krasovskii and the Lyapunov--Razumikhin approaches. In both approaches, the results are based on the same Lyapunov function of the corresponding delay free system. This fact allows us to compare the approaches directly, in the light of some illustrative examples.\\

\section*{Acknowledgments}
The authors thank the anonymous reviewers for their useful suggestions that helped us to improve the presentation of the manuscript.

\section*{Funding}
The work of Irina V. Alexandrova was supported by the Russian Science Foundation, Project 19-71-00061.
The work of Gerson Portilla and Sabine Mondi\'e was supported by project CONACYT A1-S-24796 and project SEP-CINVESTAV 155, Mexico.

\bibliographystyle{apacite}
\bibliography{mybibfile}

\section*{APPENDIX}
For the sake of completeness, we present the proofs of Theorems \ref{th:attraction_LR} and \ref{th:estimation_LR} from \citet{aleksandrov2012asymptotic,aleksandrov2014delay, aleksandrov2016asymptotic}. First, notice that, in contrast with the original works and for comparison purposes, all constants involved in Theorems \ref{th:attraction_LR} and \ref{th:estimation_LR} are computed explicitly. Second, observe that the structure of the estimates for the response in Theorem \ref{th:estimation_LR} is the same as those in \citet{aleksandrov2012asymptotic} but the constants are presented in a different way and they cover a wider class of systems \citep{aleksandrov2016asymptotic}. We begin with an auxiliary lemma, which can be found in \citet{aleksandrov2014delay} in an implicit form.
\begin{lemma}
\label{lemma_Raz}
If $\|\varphi\|_h < \Delta,$ then
$$
\|x(t,\varphi)\| \leq K (\|\varphi\|_h + m h \|\varphi\|_h^\mu),\quad t\in [0,h].
$$
\end{lemma}

\begin{proof}
Denote $S(\varphi)=\|\varphi\|_h + m h \|\varphi\|_h^\mu,$
$$
u(t) = S(\varphi) + m \int_0^t \|x(s)\|^\mu d s,
$$
and observe that
$
\|x(t)\|\leq u(t),$  $t\in[0,h].$
Further,
$$
\dot{u}(t)= m \|x(t)\|^\mu \leq m u^\mu(t),\quad u(0) = S(\varphi)<\dfrac{\kappa \delta}{K}.
$$
Integrating the last inequality, we obtain
\begin{equation}
\label{bound_u}
u(t)\leq\dfrac{S(\varphi)}{\bigl(1-(\mu-1)m S^{\mu-1}(\varphi) t\bigr)^{\frac{1}{\mu-1}}},
\end{equation}
if $t<1/\bigl((\mu-1)m S^{\mu-1}(\varphi)\bigr).$
Now, verify that
$$
\dfrac{1}{(\mu-1)m S^{\mu-1}(\varphi)}>h.
$$
Hence, bound (\ref{bound_u}) holds for $t\in[0,h].$
Next,
\begin{multline*}
\|x(t)\|\leq u(t)\leq \dfrac{S(\varphi)}{\left(1-(\mu-1)m h \left(\dfrac{\kappa \delta}{K}\right)^{\mu-1}\right)^{\frac{1}{\mu-1}}}=K S(\varphi),\quad t\in[0,h],
\end{multline*}
and the result follows. 
\end{proof}

\textbf{Proof of Theorem \ref{th:attraction_LR}}. \citep{aleksandrov2014delay}.
The Razumikhin condition \eqref{razu_condition} implies
$$\|x(\xi)\|< \left(\frac{\alpha k_1}{k_0}\right)^\frac{1}{\gamma}\|x(t)\|,\quad  \forall\, \xi\in[t-2h,t],\quad \forall\, t\geq h.$$
Differentiating $V(x(t))$ along the solutions of system \eqref{eq:delay_system} satisfying the Razumikhin condition and applying the mean value theorem, we get 
\begin{multline*}
    \frac{dV(x(t))}{dt}= \left(\frac{\partial V(x)}{\partial x}\right)^Tf(x(t),x(t))-h\left(\frac{\partial V(x)}{\partial x}\right)^T\\
    \times\int_0^1\frac{\partial f(x(t),x(t-\theta h))}{\partial x(t-\theta h)} f(x(t-\theta h),x(t-\theta h-h))d\theta\\
    \leq -\mathrm{w}\|x(t)\|^{\gamma+\mu-1} + k_4 \|x(t)\|^{\gamma+2\mu-2},
\end{multline*}
where $$k_4=2hmm_2k_2\left(\frac{\alpha k_1}{k_0}\right)^\frac{\mu}{\gamma}\left(1+\left(\frac{\alpha k_1}{k_0}\right)^\frac{\mu-1}{\gamma}\right).$$
Taking $H=\left(\mathrm{w}/k_4\right)^\frac{1}{\mu-1}$ and an arbitrary $\delta\in(0,H),$ we arrive at \eqref{eq:der_Razum} with
$k_5=\mathrm{w}-k_4\delta^{\mu-1}.$ \\
Consider an arbitrary solution $x(t)$ of system \eqref{eq:delay_system} with initial condition \mbox{$\|\varphi\|_h<\Delta$.} It follows  that $\Delta<\kappa \delta$ from equation \eqref{eq:attraction_LR}, hence
\begin{equation*}
    V(\varphi(\theta))\leq k_1\Delta^\gamma<k_0\delta^\gamma,\quad \theta\in [-h,0].
\end{equation*}
Lemma~\ref{lemma_Raz} implies $V(x(t))\leq k_1\|x(t)\|^\gamma<k_0 \delta^\gamma,$ $t\in[0,h].$ Assume that $\bar{t}>h$ is the first time instant such that the inequality is violated: $V(x(\bar{t}))=k_0 \delta^\gamma,$ $\|x(t)\|\leq\delta$ for $t\leq\bar{t}.$ Then, formula \eqref{eq:der_Razum} provides a contradiction at $t=\bar{t}$ immediately. Hence, $V(x(t))<k_0 \delta^\gamma$ for any $t\geq-h,$ which implies $\|x(t)\|<\delta$ for any $t\geq-h.$ The result now follows from \eqref{eq:der_Razum}.
\begin{remark}
The value $H$ can be taken instead of $\delta$ in equation \textup{(\ref{eq:attraction_LR})} and in $K$ as it is done in \textup{\citet{aleksandrov2014delay}}. 
\end{remark}

\textbf{Proof of Theorem~\ref{th:estimation_LR}}. \citep{aleksandrov2012asymptotic,aleksandrov2016asymptotic}
Equations \eqref{eq:bound_V} and \eqref{eq:der_Razum} imply that along the solutions satisfying (\ref{razu_condition}) and $\|x_t\|_h\leq\delta$ the following differential inequality holds 
\begin{equation}\label{deriv_V}
   \frac{dV(x(t))}{dt}\leq -\overline{d}V^\frac{\gamma+\mu-1}{\gamma}(x(t)),\quad \overline{d}=k_5k_1^{-\frac{\gamma+\mu-1}{\gamma}}.
\end{equation}
Lemma~\ref{lemma_Raz} provides an initial condition for this inequality:
$$V(x(h))\leq k_1K^\gamma(\|\varphi\|_h+mh\|\varphi\|_h^\mu)^\gamma.$$
Introduce a parameter $\rho,$ which satisfies the following three conditions:
\begin{gather}
    0<\rho<\overline{d},\notag\\
    1+2h\rho\frac{\mu-1}{\gamma}k_0^\frac{\mu-1}{\gamma}\delta^{\mu-1}<\alpha^\frac{\mu-1}{\gamma},\label{firs_condition}\\
    \label{second_condition}
    1-\rho\frac{\mu-1}{\gamma}k_1^\frac{\mu-1}{\gamma}K^{\mu-1}h\Delta^{\mu-1}>0.
\end{gather}
Then, the differential equation 
\begin{equation}\label{eq:equa_comparison}
    \dot{z}(t)=-\rho z^{\frac{\gamma+\mu-1}{\gamma}}(t)
\end{equation} 
with the initial condition $z(h)=z_0,$ where $$z_0 = k_1K^\gamma(1+mh\Delta^{\mu-1})^\gamma\|\varphi\|_h^\gamma,\quad z_0<k_0\delta^\gamma,$$ can be treated as a comparison equation for \eqref{deriv_V}, if $t\geq h.$ A solution to this initial-value problem is 
\begin{equation*}
    z(t)=\frac{z_0}{\left(1+\rho\dfrac{\mu-1}{\gamma} z_0^\frac{\mu-1}{\gamma}(t-h)\right)^\frac{\gamma}{\mu-1}}.
\end{equation*}
Let us show that the solution $z(t)$ satisfies the Razumikhin condition \eqref{razu_condition}. Consider $t\geq h$ and $\xi$ such that $-2h\leq \xi \leq 0$ and $t+\xi\geq h,$ then condition \eqref{razu_condition} requires $z(t+\xi)<\alpha z(t)$, equivalently, 
\begin{equation}\label{eq:prove_rho}
    g(t)=\frac{1+k(t-h)}{1+k(t+\xi-h)}<\alpha^\frac{\mu-1}{\gamma},
\end{equation}
where $k=\dfrac{\rho(\mu-1)}{\gamma} z_0^\frac{\mu-1}{\gamma}$.
Note that
\begin{equation*}
    0\leq g(t)=1-\frac{k\xi}{1+k(t+\xi-h)}\leq 1-k\xi\leq 1+2kh.
\end{equation*}
Since $k< \rho\dfrac{\mu-1}{\gamma} k_0^\frac{\mu-1}{\gamma}\delta^{\mu-1}$, we have that condition \eqref{firs_condition} implies \eqref{eq:prove_rho}.
Hence, function
$z(t)$ satisfies the Razumikhin condition for all $t\geq h$. \\
Following the general idea of the Razumikhin framework, we conclude that $V(x(t))\leq z(t)$, $\,t\geq h$, for all solutions $x(t)$ with $\|\varphi\|_h<\Delta.$
This implies that the solutions with $\|\varphi\|_h<\Delta$ admit the following bound:
\begin{equation*}
    \|x(t,\varphi)\|\leq \frac{A\|\varphi\|_h}{(1+B\|\varphi\|_h^{\mu-1}(t-h))^\frac{1}{\mu-1}},\quad t\geq h,
\end{equation*}
where
\begin{align*}
    A&=\frac{K(1+mh\Delta^{\mu-1})}{\kappa}=\frac{\delta}{\Delta},\\
    B&=\rho\frac{\mu-1}{\gamma}k_1^\frac{\mu-1}{\gamma}\left(K(1+mh\Delta^{\mu-1})\right)^{\mu-1}.
\end{align*}
Furthermore, since $1-Bh\Delta^{\mu-1}>0$ due to condition \eqref{second_condition}, we have
\begin{gather*}
    \|x(t,\varphi)\|\leq \frac{A\|\varphi\|_h}{(1-B\|\varphi\|_h^{\mu-1}h+B\|\varphi\|_h^{\mu-1}t)^\frac{1}{\mu-1}}\leq \frac{\tilde{c}_1\|\varphi\|_h}{(1+\tilde{c}_2\|\varphi\|_h^{\mu-1}t)^\frac{1}{\mu-1}},
\end{gather*}
if $t\geq h,$ where
\begin{equation*}
    \tilde{c}_1=\frac{A}{(1-Bh\Delta^{\mu-1})^\frac{1}{\mu-1}},\quad \tilde{c}_2=\frac{B}{1-Bh\Delta^{\mu-1}}.
\end{equation*}

It remains to obtain the estimate for $t\in [0,h]$. According to Lemma~\ref{lemma_Raz},
\begin{equation*}
    \|x(t,\varphi)\|\leq A_1\|\varphi\|_h,\quad t\in [0,h],
\end{equation*}
where $A_1=K(1+mh\Delta^{\mu-1}).$
Notice that
\begin{equation*}
    \dfrac{(1+\tilde{c}_2\Delta^{\mu-1}h)^\frac{1}{\mu-1}A_1}{\tilde{c}_1}=\kappa\leq1.
\end{equation*}
Hence,
\begin{equation*}
    \|x(t,\varphi)\|\leq  \frac{\tilde{c}_1\|\varphi\|_h}{(1+\tilde{c}_2\Delta^{\mu-1}h)^\frac{1}{\mu-1}}
    \leq \frac{\tilde{c}_1\|\varphi\|_h}{(1+\tilde{c}_2\|\varphi\|_h^{\mu-1}t)^\frac{1}{\mu-1}},\quad t\in[0,h].
\end{equation*}
Therefore, the required bound \eqref{eq:final_est_LR} holds for all $t\geq0,$ and the proof is complete.

\end{document}